\documentclass[10pt,a4paper]{article}
\pdfoutput=1
\usepackage{a4}
\usepackage{hyperref}
\usepackage{natbib}
\usepackage{subfiles}

\usepackage{float}
\usepackage{setspace}
\usepackage{url}
\usepackage{enumitem}
\usepackage{amstext}
\usepackage{xcolor}
\usepackage{amsfonts}
\usepackage{amssymb}
\usepackage{graphics}
\usepackage{graphicx}
\usepackage{pstricks,pst-node,pst-tree}
\usepackage{enumerate}
\usepackage{tikz} 
\usepackage{cutwin}
\usepackage{xcolor}
\usepackage{colortbl}
\usepackage{algorithm,algorithmic}
\usepackage{multirow}
\usepackage{amsthm}
\usepackage{glossaries}
\usepackage{array,arydshln}
\usetikzlibrary{shapes,matrix}
\usepackage{comment}
\usepackage{mathrsfs}
\usepackage{authblk}
\usepackage{wrapfig}
\usepackage[font=footnotesize]{caption}
\usepackage{stackengine}
\usepackage{pifont}

\usepackage{xr}
\makeatletter

\newcommand*{\addFileDependency}[1]{
\typeout{(#1)}
\@addtofilelist{#1}
\IfFileExists{#1}{}{\typeout{No file #1.}}
}\makeatother

\newcommand{\cmark}{\ding{51}}%
\newcommand{\xmark}{\ding{55}}%

\newtheorem{theorem}{Theorem}
\newtheorem{proposition}{Proposition}

\newtheorem{remark}{Remark}

\newcommand  {\mps}{m.p.\;}
\newcommand  {\jmps}{$^J$m.p.\;}

\newcommand{\N}{\mathbb{N}}

\newcommand{\B}{\mathbb{B}}

\newcommand\stackrqarrow[2]{%
    \mathrel{\stackunder[2pt]{\stackon[4pt]{$\leadsto$}{$\scriptscriptstyle#1$}}{%
            $\scriptscriptstyle#2$}}}

\definecolor{sable}{RGB}{153,51,0}
\definecolor{macouleur}{RGB}{180,0,0}
\definecolor{ForestGreen}{RGB}{0,180,0}
\definecolor{Aquamarine}{rgb}{0.5, 1.0, 0.83}
\definecolor{aqua}{rgb}{0.0, 1.0, 1.0}
\definecolor{darkcyan}{rgb}{0.0, 0.55, 0.55}
\definecolor{darkpastelpurple}{rgb}{0.59, 0.44, 0.84}
\definecolor{ao(english)}{rgb}{0.0, 0.5, 0.0}
\definecolor{atomictangerine}{rgb}{1.0, 0.6, 0.4}
\definecolor{amaranth}{rgb}{0.9, 0.17, 0.31}

\title{Refining Boolean models with the partial most permissive scheme}

\author[1,2]{Nadine Ben Boina}
\author[1]{Brigitte Mossé}
\author[2,3,4]{ Anaïs Baudot}
\author[1]{Elisabeth Remy}
\affil[1]{Aix Marseille Univ, CNRS, I2M (UMR 7373), Turing Center for Living systems, Marseille, France}
\affil[2]{Aix Marseille Univ, INSERM, MMG, Marseille, France}
\affil[3]{CNRS, Marseille, France}
\affil[4]{Barcelona Supercomputing Centre, Barcelona, Spain}
\begin{document}
\maketitle

\paragraph{Abstract}
\paragraph{Motivation} 
In systems biology, modelling strategies aim to decode how molecular components interact to generate dynamical behaviour. Boolean modelling is more and more used, but the description of the dynamics from two-levels components may be too limited to capture certain dynamical properties. 
Multivalued logical models can overcome this limitation by allowing more than two levels for each component.  However, multivaluing a Boolean model is challenging.

\paragraph{Results} 


We present MRBM, a method for efficiently identifying the components of a Boolean model to be multivalued in order to capture specific fixed-point reachabilities in the asynchronous dynamics. To this goal, we defined a new updating scheme locating reachability properties in the most permissive dynamics. MRBM is supported by mathematical demonstrations and illustrated on a toy model and on two models of stem cell differentiation.

\paragraph{Availability and Implementation} Data and code available at: 

https://github.com/NdnBnBn/MRBM.git.
\paragraph{Contact} elisabeth.remy@univ-amu.fr
\clearpage

\section{Introduction}
Systems biology aims to understand how the interactions between biological components such as molecules, cells, tissues, or organs, give rise to specific phenotypes and behaviours \cite{bruggeman_nature_2007}. Mathematical models, integrating data from experiments and literature to represent biological systems and facilitate their study, play a pivotal role in systems biology. They have been instrumental in studying a broad spectrum of biological systems, ranging from signal transduction \cite{heinrich_mathematical_2002} to complex developmental processes \cite{tomlin_biology_2007}. Importantly, mathematical models reduce the cost and time associated with \textit{in vivo} or \textit{in vitro} experiments. 

Various mathematical frameworks are used to model biological interaction networks \cite{chen_classic_2010, ay_mathematical_2011, torres_mathematical_2015}. We focus here on logical modeling, a formalism that qualitatively captures the main dynamical properties that explain the overall behaviour of the system, without requiring precise parameters. In this formalism, the activity levels of the components of a biological system are represented as discrete variables.
Logical models encompass both Boolean models (BMs) and multivalued models. In BMs, the activity level of the components are represented as binary variables (inactive/active) \cite{thomas_boolean_1973}. However, BMs may not be sufficient to capture the complex dynamics of biological systems. Multivalued models can address this limitation by allowing some components of the model to have more than two levels of activity, while retaining the conceptual simplicity of logical modeling \cite{thomas_logical_1978}. 
When a multivalued model is derived from an existing BM, without adding new components nor modifying the effects of the regulations (activation or inhibition), we refer to this multivalued model as a multivalued refinement of the BM \cite{pauleve_reconciling_2020}.

The dynamics of logical models is driven by logical functions and an updating scheme. The logical functions specify the conditions of activation of each component, while the updating scheme dictates how the system is updated. The choice of updating scheme is crucial as it impacts the dynamics of the system. Two updating schemes are commonly used, the synchronous and the asynchronous updating schemes\footnote{sometimes denoted fully asynchronous} \cite{garg_synchronous_2008}. Under the synchronous scheme, the activity levels of all the components are updated simultaneously, resulting in a deterministic dynamics. Under the asynchronous scheme, a single component's activity level is updated at a time. The resulting dynamics is not deterministic, as some states may have several successors. The asynchronous scheme is often preferred in modeling biological systems. Indeed, the asynchronous scheme enables the emergence of complex behaviours that are not revealed with the synchronous scheme \cite{de_maria_boolean_2022}. The most permissive updating scheme (\mps scheme), recently introduced in \cite{pauleve_reconciling_2020}, is a non-deterministic scheme that only applies to BMs. The \mps scheme adds to each component two intermediate activity levels that reflect the transition between the two Boolean levels. In these intermediate levels, the components are considered by their targets to be both active and inactive. The update is then asynchronous. The dynamics of a BM under the \mps scheme captures the dynamics obtained under the synchronous and asynchronous schemes. Interestingly, it also captures the dynamics of any multivalued refinement of a BM \cite{pauleve_reconciling_2020}.  

Attractors, which represent the long-term behaviours of a model, are key features of the dynamics. Once identified, they can be associated with distinct phenotypic readouts \cite{schwab_concepts_2020}. The existence of a trajectory from a given state to a specific attractor proves the reachability of the attractor from this state. The number of attractors present in the dynamics, and their reachability, are impacted by the choice of the updating scheme. Indeed, the reachability properties are easier to achieve in the \mps dynamics, compared with the synchronous and asynchronous dynamics, or in any multivalued refinement of the BM (see \cite{pauleve_reconciling_2020}).

We introduce here MRBM (Multivalued Refinement of Boolean Model), a method which aims at identifying components to be multivalued in a refinement of a BM in order to provide the desired reachabilities within the asynchronous dynamics. MRBM utilizes the partial \mps scheme, a new updating scheme that we are proposing here (Section \ref{partialmp}) and that we adapted from the \mps scheme. In the partial \mps, only a subset of the model's components is updated using the \mps scheme, the remaining components being updated with the asynchronous scheme. The resulting dynamics helps to pinpoint the components of the BM that need to be multivalued. In this manuscript, we give a detailed description of the method MRBM. We illustrated its potential with a Toy model (Section \ref{mrbm}). We then used the MRBM method to propose multivalued refinements of two BMs of stem cell differentiation. The first model (Section \ref{sec: HSC}) describes the aging effects on hematopoietic stem cell differentiation. The second model (Section \ref{sec: AT}) explores the differentiation of stem cells in the root of \emph{Arabidopsis thaliana}. In both cases, we found a multivalued refinements which display reachabilities that were not present in the BMs.

\section{Method}
\label{sec:Method}
\subsection{Boolean Modeling} \label{sec:BM}
Set $\B = \{0,1\}$ and $n$ an integer $>0$. 
A {\it Boolean Model} (BM) of dimension $n$ is a map $f = (f_1, \dots, f_n) \;:\; \B^n\;\longrightarrow \B^n\,,$ where each Boolean state  $x = (x_1, \dots, x_n) \in \B^n$ specifies the activity level of $n$ components of the model $g_1$, ..., $g_n$. The values of the logical regulatory function $f_i(x) \in \B$ provide the target levels of the $i$-th component.
This model is associated with a {\it regulatory graph}, i.e. a directed signed graph. The $n$ nodes $g_1, ... , g_n$ of the graph are the components of the model. The edges between nodes represent the regulations between components (activation or inhibition) (Figure \ref{fig:Toy}A). 

\begin{figure}[h!]
    \centering
    \includegraphics[page=1, width=\linewidth]{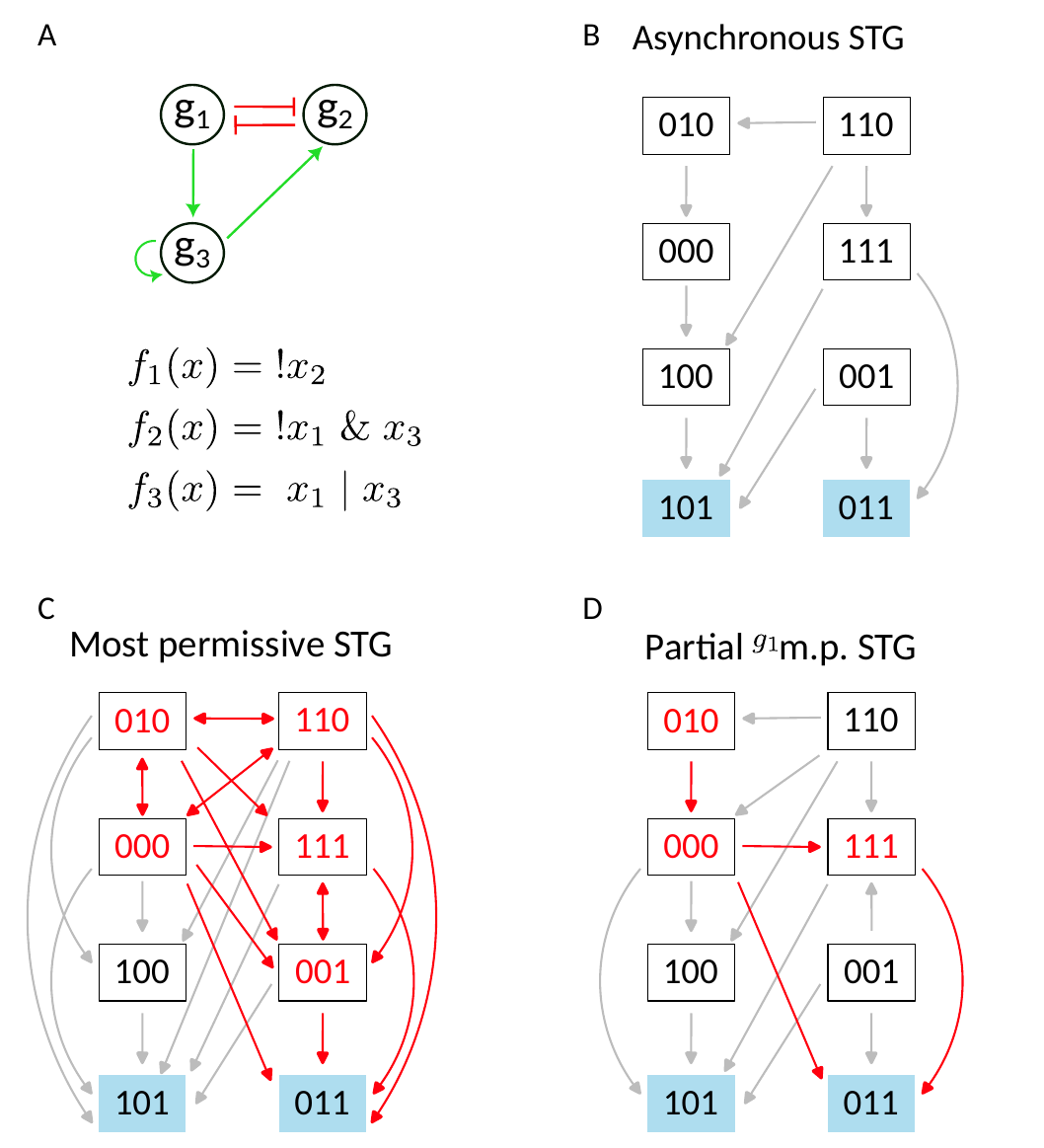}
    \caption{\textbf{Toy Boolean Model (BM)}. ({\bf A}) Regulatory graph associated with a BM $f$ composed of three components $g_1, g_2, g_3$, five regulatory interactions (red edges for inhibition, green edges for activation), and logical regulatory functions (operators $"\&", "|","!"$ stand for {\sc and, or, not} respectively). The values of the logical regulatory function $f_i(x) \in \mathbb{B}$ provide the target levels of the $i$-th component. ({\bf B}) State Transition Graph (STG) representing the Toy model's dynamics under the asynchronous scheme. Each node represents a state $x = (x_1, x_2, x_3)$ of the model and each edge represents a transition between two consecutive states. The attractors (fixed points) are colored in blue. ({\bf C}) Most permissive and ({\bf D}) Partial most permissive applied on gene $g_1$ ( $^{g_1}$\mps) STG. For the sake of visualization, we do not show the \mps states. The trajectories from the state $010$ to the state $011$ are highlighted in red.}
    \label{fig:Toy}
\end{figure}

The dynamics of the model is represented by the {\it State Transition Graph} (STG). The STG is a directed graph $(\B^n, T)$ in which $\B^n$ is the set of nodes that represent the states of the model, and $T \subset \B^n\times \B^n$ is the set of edges that represent the transitions between states. Within the {\it asynchronous scheme}, the activity level of only one component can be updated at a time according to the function $f$  (Figure \ref{fig:Toy}B). 
More precisely, there is a transition from the state $x = (x_1, \dots, x_n)$ to the state $y = (y_1, \dots, y_n)$ if there exists $i_0 \in \{1, \dots, n\}$ such that $x_{i_0} \neq f_{i_0}(x)$ and
\begin{align*}
& y_{i_0} = x_{i_0} + \text{sign}(f_{i_0}(x) - x_{i_0}), \\
& y_j = x_j,\;\; \text{for}\;\;j \neq i_0.
\end{align*}
Note that the asynchronous scheme provides a non-deterministic dynamics, as each state can have several successors.

A transition from state $x$ to state $y$ using the asynchronous scheme is denoted by $x \xrightarrow[\text{asyn}]{f} y$. A {\it trajectory} from state $x$ to state $y$ is a sequence of transitions $x \xrightarrow[\text{asyn}]{f} x' \xrightarrow[\text{asyn}]{f} \dots \xrightarrow[\text{asyn}]{f} y$. 
A state $y$ is {\it reachable} from $x$ if there exists a trajectory from $x$ to $y$. We denote this reachability by $x\stackrqarrow{f}{asyn}y$.

The {\it attractors} of the STG are the terminal strongly connected components (SCCs) of the graph, gathering the asymptotic states of the model. We distinguish two types of attractors, fixed points and cyclic attractors. Fixed points are terminal SCCs of size $1$, whereas cyclic attractors are terminal SCCs of size $\geq 2$. Given an attractor $\omega$, the basin of attraction $\mathcal{B}^f_{\text{asyn}}(\omega)$ is the set of states from which a state of $\omega$ is reachable.

The above concepts on trajectories and attractors can be adapted to multivalued models and other updating schemes.

\subsection{Multivalued Refinements of Boolean Models}\label{mv_refinements}
A natural extension of BMs is to have multivalued variables attached to some components, enabling these components to have more than two levels of activity.\\
Given $m_j\in \N ^*$ for $j\in \{1,\dots,n\}$, with at least one $m_j>1$, let $X= \prod_{j=1}^n\;\{0,1,\dots, m_j\}$.  A \textit{multivalued model} on $X$ is a map $h=(h_1, \dots , h_n):X\rightarrow X$ satisfying the condition $h_j(x)-x_j \in \{-1,0,+1\}$, for each $x=(x_1,\dots,x_n) \in  X$ and $j\in \{1,\dots,n\}$.

\medskip
Let us consider $f=(f_1, \dots, f_n)$ a BM of dimension $n$. Following \cite{pauleve_reconciling_2020}, and given $X= \prod_{j=1}^n\;\{0,1,\dots, m_j\}$ as above, we introduce for $x\in X$ the set \\
$\alpha(x)= \{x' \in \B ^n \;;\; \forall j\in \{1,\dots,n\},(x_j=0\Rightarrow x'_j=0)\textrm{ and }(x_j=m_j\Rightarrow x'_j=1)\}$.

A multivalued model $h$ on $X$ is said to be a \textit{refinement of $f$} if, for any $x=(x_1,\dots,x_n) \in  X$ and $j\in \{1,\dots,n\}$,
\begin{align*}
&h_j(x)<x_j \Longrightarrow \exists x'\in \alpha(x), f_j(x')<x'_j, \\
&h_j(x)>x_j \Longrightarrow \exists x'\in \alpha(x), f_j(x')>x'_j.
\end{align*}
\noindent
This means that, in a multivalued refinement, the level of activity of a component may decrease (resp. increase) at a given state only if its value goes to $0$ (resp. $1$) under $f$ at some Boolean state of $\alpha(x)$. 

Starting from $f$, we use the following approach to parameterize a multivalued refinement, i.e. to define the function $h$. For each multivalued component $g_j$, we consider the logical formulas of its targets $g_k$ (this means that $g_k$ is regulated by $g_j$). For each conjunctive clause (i.e. conjunctive logical assertion) of these formulas, we choose a threshold $s$ in $\{0, \dots,m_j-1\}$. This threshold specifies that the regulation of $g_k$ occurs when $x_j \in \{s+1, \dots,m_j\}$. The threshold $s$ is chosen according to the target and to the clause (see details in Supplementary Material Section \ref{AnnexA}). 

In the remaining sections of this manuscript, all the refinements we consider are built as described above.
In the Results section, we described these refinements $h$ through the functions $\mathcal{H}_j$ of $x\in X$ with values in $\{-1,+1\}$. These functions $\mathcal{H}_j$ indicate whether $x_j$ is called to decrease or to increase under $h$ (when possible, i.e. remaining between $0$ and $m_j$).

Note that in most cases, given a BM $f$, if there exists an asynchronous path from $x \in \B^n$ to $y \in \B^n$, a corresponding path also exists in the asynchronous dynamics of any refinement of $f$. Moreover, for any fixed point in the asynchronous dynamics of $f$, a corresponding fixed point exists in the asynchronous dynamics of any refinement of $f$ belonging to $\prod_{j=1}^n\;\{0, m_j\}$, and conversely.
There is one exception. When any multivalued component is self-inhibited, these results are irreversibly defeated: reachability properties are modified, and even the number and nature of attractors may change (Supplementary Material Section \ref{AnnexB}).

Finally, it should be noted that there is no need to multivalue an output (component without target) in a refinement. Indeed, the multivaluation of an output would have no consequences in terms of reachability properties.

\subsection{Most Permissive Scheme}\label{ssec:m.p.}
The \textit{most permissive scheme} (\mps scheme) proposed in \cite{pauleve_reconciling_2020} is a non-deterministic updating scheme where each component can take four activity levels: $0$, $1$, $i$, and $d$. We set $X_{m.p.}=\{0,1,i,d\}^n$. The increasing and decreasing levels $i$ and $d$ refer to intermediate levels. A component must go through these intermediate levels to reach its target level, either $0$ or $1$. A component at an intermediate level can be considered by its targets as being at either level $0$ or level $1$.

\medskip
To each state $x$ of $X_{m.p.}$ is associated a set of Boolean states: $$\gamma (x)=\{x'\in \B ^n; \forall j \in \{1,\dots,n\}, (x_j=0\Rightarrow x'_j=0)\textrm{ and } (x_j=1\Rightarrow x'_j=1)\}.$$
Then, given a BM $f$, there is a transition from a state $x \in X_{m.p.}$ to a state $y \in X_{m.p.}$ if there exists $j_0 \in \{1,...,n\}$ such that one of the following occurs
\[  
\begin{array}{ll}
- \;x_{j_0}\in \{0, d\}   \;\;\text{and}\;\;   \exists x' \in \gamma (x) \text{ such that } f_{j_0} (x') = 1   \;\;\;\;\text{and}\; \;   y_{j_0}=i, \\
-\; x_{j_0}\in \{1, i\}    \;\;\text{and}\;\;   \exists x' \in \gamma (x) \text{ such that } f_{j_0} (x') = 0   \;\;\;\;\text{and}\;\;    y_{j_0}=d, \\
-\; x_{j_0}=i                \; \;\;\;\text{and}\;\;    y_{j_0} = 1, \\
- \;x_{j_0}=d                 \;\;\;\;\text{and}\; \;   y_{j_0} = 0, \\
\end{array}
\]
\medskip
and $y_j = x_j, \text{\;for\;} j \neq j_0$.

All states with at least one component activity level equal to $i$ or $d$ will be referred to as \textit{\mps states}, and other states as \textit{Boolean states}.
For the sake of clarity, when representing the \mps STG we do not show the \mps states (Figure \ref{fig:Toy}C). 

An interesting property emerges from the \mps scheme: given a BM, the \mps scheme captures all the trajectories present in any of its multivalued refinements. This completeness property of the \mps is described in \cite{pauleve_reconciling_2020}). Consequently, if a reachability is not achievable using the \mps scheme, then no multivalued refinement will present this reachability. 

\subsection{Partial Most Permissive Scheme} \label{partialmp}

We introduce here the \textit{partial \mps scheme}, derived from the \mps scheme. Each of these new updating schemes is associated with the choice of a subset $J$ of $\{1, \dots , n\}$. We then assign the four possible levels $0$, $1$, $i$, and $d$ to components $g_j$ such that $j\in J$. We refer to these chosen components as "\mps components", and we keep the Boolean levels for all other components. 
So we consider the set of states $ X_{^Jm.p.}=\prod_{j=1}^{n} \mathcal{A}_j$, where $\mathcal{A}_j = \{0,1,i,d\}$ if $j\in J$, and $\mathcal{A}_j
=\{0,1\} $ otherwise.

Given a BM $f$ and $J$ as above, we then define the \textit{partial \jmps scheme} on $ X_{^Jm.p.}$ in the following way:
there is a transition from a state $x \in X_{^Jm.p.}$ to a state $y \in X_{^Jm.p.}$ if there exists $j_0 \in \{1,...,n\}$ such that one of the following occurs
\[  
\begin{array}{ll}
-\; x_{j_0}\in \{0, d\}   \;\;\text{and}\;\;   \exists x' \in \gamma (x) \text{ such that } f_{j_0} (x') = 1   \;\;\;\;\text{and}\; \;   y_{j_0}=i, \\
-\; x_{j_0}\in \{1, i\}    \;\;\text{and}\;\;   \exists x' \in \gamma (x) \text{ such that } f_{j_0} (x') = 0   \;\;\;\;\text{and}\;\;    y_{j_0}=d, \\
-\; x_{j_0}=i                \; \;\;\;\text{and}\;\;    y_{j_0} = 1, \\
- \;x_{j_0}=d                 \;\;\;\;\text{and}\; \;   y_{j_0} = 0, \\
-\;x_{j_0}=0 \;\;\text{and}\;\; j_0 \notin J  \;\;\text{and}\;\;   \exists x' \in \gamma (x) \text{ such that } f_{j_0} (x') = 1   \;\;\;\;\text{and}\; \;   y_{j_0}=1, \\
-\; x_{j_0}=1 \;\;\text{and}\;\;  j_0 \notin J \;\;\text{and}\;\;   \exists x' \in \gamma (x) \text{ such that } f_{j_0} (x') = 0   \;\;\;\;\text{and}\;\;    y_{j_0}=0, \\
\end{array}
\]
\medskip
and $y_j = x_j, \text{\;for\;} j \neq j_0$.



Note that if the size of the set $J$ is equal to $n$ (the dimension of the model), then the partial \jmps scheme is the \mps scheme. If $J$ is empty, the partial \jmps scheme is the asynchronous scheme.

The completeness property of the \mps scheme can be generalized to the partial \mps scheme: given a BM $f$ and a set $J$, the \jmps dynamics encompasses the dynamics of any multivalued refinement of $f$ with $\{g_j\;;\;j\in J\}$ as the set of multivalued components (see proof in Supplementary Material Section \ref{AnnexC}). 

We can emphasize the following properties : 
\begin{itemize}
    \item If $J_1$ and $J_2$ are subsets of $\{1,\dots,n\}$ such that $J_1$ is contained in $J_2$, then $X_{^{J_1}m.p.}$ is contained in $X_{^{J_2}m.p.}$; moreover, if there exists a trajectory between two states in the  $^{J_1}m.p.$ dynamics of $f$, the same holds in the $^{J_2}m.p.$ dynamics. 
    In particular, if there exists a trajectory between two Boolean states in the asynchronous dynamics of $f$, the same property holds in any partial m.p. dynamics.
    \item The asynchronous dynamics and the partial $m.p.$ dynamics of $f$ share the same fixed points.
\end{itemize}
For the sake of simplicity, in the Results section, we will identify the set $J$ as the set of \mps components.

\subsection{Assessment of fixed points reachability property}\label{sec:reachability}
Assessing the reachability of attractors involves determining whether specific states can reach an attractor, or assessing the entirety of states that can reach an attractor (i.e. characterizing the basin of attraction of an attractor). In this section, we will describe the methods we used to identify the existence of reachabilities in the dynamics of a logical model and to compute the basins of attraction.

The MRBM approach is based on comparing the reachability properties of fixed points across the dynamics obtained with different updating schemes for the same BM.
The fixed points of a BM do not change with the updating scheme. However, the basins of attraction of these fixed points differ according to the updating scheme, reflecting changes in reachability. So, to identify these changes in reachability, we compare the basins of attraction. The state spaces are not the same depending on the choice of update: the \mps states are not present in the Boolean asynchronous dynamics. In order to make a relevant comparison between \mps, partial \mps  and asynchronous dynamics, we have only considered Boolean states. 
%

In the same way, there is a one-to-one correspondence between the fixed points of a BM and those of any of its multivalued refinements, provided there is no self-inhibited multivalued component (see Section \ref{mv_refinements}). To compare the basins of attractions in a BM with the basins in a multivalued refinement, we only took into account the states of the subspace $\prod_{j=1}^n\;\{0, m_j\}$. This subspace is in correspondance with the set $\B^n = \prod_{j=1}^n\;\{0, 1\}$. 

\subsubsection{Reachability Identification}

The tools available to identify reachabilities do not support all the possible updating schemes. In this context, for BMs updated with the asynchronous or partial \mps schemes, we verified reachability with model checking \cite{baier_principles_2008}. To do so, we used the NuSMV model checker software extension \cite{cimatti_nusmv_2002} provided in the PyBoolNet Python package \cite{klarner_basins_2020} and we expressed the attractor reachability property to be verified using computational tree logic (CTL) formulas, as described in \cite{ baier_principles_2008, klarner_basins_2020}. 
For BMs updated with the \mps scheme, we used the Python package mpbn \cite{pauleve_reconciling_2020}. This package is based on Answer-Set Reprogramming \cite{gebser_answer_2013} and on the solver Clingo \cite{gebser_clingo_2014}.

\subsubsection{Basins of Attraction Identification} \label{basin of attraction sizes}
Given $f$ a logical model (Boolean or any multivalued refinement) and its asynchronous dynamics, we use model checking to compute the basin of attraction of a fixed point $\omega$. This is done by identifying, for any state $x$, whether there exists a path from  $x$ to $\omega$. 

In the partial \mps scheme, basins of attraction are computed in the same way. There is however a preliminary step to generate the partial \mps dynamics of a BM using the bioLQM toolkit \cite{naldi_biolqm_2018}. 
In the \mps scheme, the basin of attraction are computed  using the Python package mpbn \cite{pauleve_reconciling_2020}.

We use the size of the basins of attraction of fixed points as a measure of attractor reachability. 
For a logical model $f$, we define the size of the basin of attraction of $\omega$, denoted by $\#\mathcal{B}^f(\omega)$, as the number of Boolean states (or equivalent states for a multivalued refinement) present in the basin. We report the sizes of the basins as percentages relative to the sizes of the Boolean states space ($\frac{\#\mathcal{B}^f(\omega)}{2^n} \times 100 \% $).

\subsection{Availability and Implementation}
\label{sec:implementation}

The MRBM method and the BMs used in this manuscript are available on GitHub at: https://github.com/NdnBnBn/MRBM.git. This repository contains the scripts necessary to run the MRBM method within the MRBM directory. Three supplementary directories contain the toy model and the two stem cell differentiation models presented in the Results section. Each of these directories contains the original BM files, the output generated after running MRBM, and the multivalued refinements described in the results section. The method relies on the Colomoto Docker environment \cite{naldi_colomoto_2018}.

\section{Results} \label{sec:Results}

\subsection{The MRBM Method: defining Multivalued Refinements of Boolean Models} \label{mrbm}

\subsubsection{Description of the MRBM method} \label{description}
Let $f$ a BM such that its asynchronous dynamics displays several fixed points and no cyclical attractor. 
Suppose that a reachability property ${\cal P}$ is satisfied in the \mps dynamics of the BM $f$, but not in the asynchronous dynamics. The MRBM method aims to identify multivalued refinements of $f$ such that the asynchronous dynamics exhibits the reachability property ${\cal P}$.

We know that ${\cal P}$ is satisfied when the most permissive update is applied to the whole system. The question behind the MRBM method is: can ${\cal P}$ be satisfied by applying the most permissive update to only a subset of the components of the system (the other components being asynchronously updated)? If yes, this subset will contain the components that need to be multivalued for the asynchronous dynamics to satisfy ${\cal P}$.
Thus, the MRBM method is based  on the use of the partial \jmps dynamics of the BM $f$ (defined in section \ref{partialmp}), where $J$ is the set of components to be identified.

Since self-inhibited and output components are excluded from the sets $J$ (see Section \ref{mv_refinements}), we will only consider the remaining components, called the admissible components. We initiate the method by examining whether there are sets $J$ of admissible components of size $1$ that satisfy the reachability property ${\cal P}$ in the \jmps dynamics of $f$. If there are no such set $J$ of size $1$, we increment the set size by one and reexamine. This iterative process continues until we identify set(s) of size $l \leq n$ that fulfill the reachability property ${\cal P}$.

The next step of the MRBM method is to build a multivalued refinement with the components $g_j\in J$ multivalued and the other components remaining Boolean. To this goal, the logical functions of each multivalued node and its targets need to be parameterised. We favor refinements where the maximum activity of the multivalued components is as low as possible (principle of parsimony). Hence, we decided to start this parameterisation with 3 levels of activity (i.e. $m_j = 2$ if $g_j\in J$) and to increase the maximal level of activity only if necessary. Finding the multivalued functions of the refinement consists in assigning a threshold (between $1$ and $m_j$) to each of the outgoing regulations of the multivalued components.
Defining multivalued logical functions such that the asynchronous dynamics of the model satisfies the reachability property ${\cal P}$ is not an easy task and several strategies exist. One strategy, that we call exhaustive strategy, assesses all possible logical functions, i.e. assesses all possible threshold assignations to the clauses of these logical functions (see Supplementary Material Section \ref{AnnexB}). The main issue with this strategy is the combinatorial explosion of possibilities. Another strategy, that we call the \textit{ad hoc} strategy, is to identify the sequence of regulations occurring in a path of interest in the \jmps dynamics, and deduce the thresholds of regulation from it. Importantly, in both strategies, the proposed multivalued refinements can be supported or validated by the use of available biological data. This is specially relevant when using exhaustive strategy, as this process can actually generate multiple solutions. These solutions can be filtered using biological data to retain the most biologically relevant refinements.
Once the logical functions have been defined, model checking approaches are used to verify the reachability property ${\cal P}$ (see Section \ref{sec:reachability}).

\subsubsection{Illustration of the MRBM method on a Toy Boolean Model} \label{TMex}
To illustrate the MRBM method, we considered the Toy BM $f$ of Figure \ref{fig:Toy}A, which describes regulations between three components $g_1$, $g_2$, $g_3$. Both the asynchronous and the \mps dynamics of the BM contain two fixed points and no cyclic attractor. We will study a specific reachability property ${\cal P}_1$, and a set of reachability properties ${\cal P}_2$. 

First, we considered a specific reachability property ${\cal P}_1$, the existence of a trajectory from the state $010$ to the fixed point $011$. We can observe that there are trajectories from $010$ to $011$ in the \mps STG (in red in Figure \ref{fig:Toy}C), but not in the asynchronous STG (Figure \ref{fig:Toy}B). Our aim was to identify a multivalued refinement of the Toy BM for which the reachability property ${\cal P}_1$ is satisfied in the asynchronous dynamics. To do this, we used the MRBM method. Following the method described in Section \ref{description}, we first searched for sets of admissible components $J$ of size one for which the reachability $010\stackrqarrow{f}{^Jm.p.}011$ is satisfied in the partial \jmps dynamics of the BM. Among all the sets of size $1$ $\{g_1\}$, $\{g_2\}$ and $\{g_3\}$, only the set $J = \{g_1\}$ fulfils the required property (Figure \ref{fig:Toy}D).
\begin{figure}[h!]
    \centering
    \includegraphics[page=2,width=\linewidth]{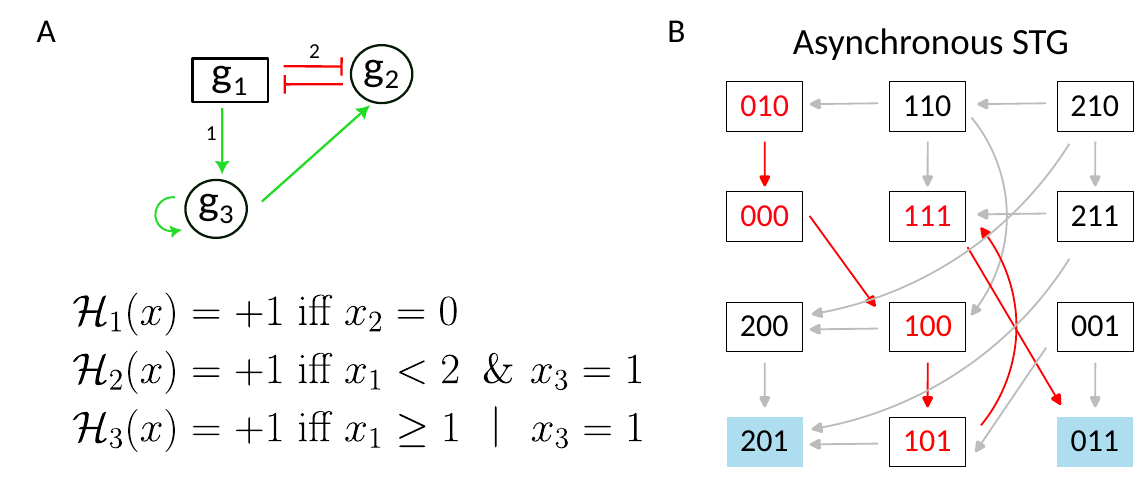}
    \caption{\textbf{Multivalued Refinement of the Toy model}. (\textbf{A}) Regulatory graph of the multivalued refinement of the BM of Figure \ref{fig:Toy}A, and associated logical functions. The rectangle represents the multivalued component $g_1$ ($m_1 = 2$); outgoing edges from $g_1$ have labels that specify the threshold levels from which the regulations occur. (\textbf{B}) Asynchronous STG of the multivalued refinement $h$. Each node represents a state $(x_1, x_2, x_3)$. Trajectories representing the reachability $010\stackrqarrow{h}{asyn.}011$ are highlighted in red.}
    \label{fig:ToyMulti}
\end{figure}
Therefore, component $g_1$ has to be multivalued. To define the multivalued functions of $g_1$ and its targets, we applied the {\it ad-hoc} strategy. We investigated the shortest trajectory from state $010$ to the fixed point $011$ in the $^{g_1}$\mps dynamics of the Toy model :$\ 010\xrightarrow{}0\textcolor{red}00\xrightarrow{}\textcolor{red}i00\xrightarrow{}i0\textcolor{red}1\xrightarrow{}i\textcolor{red}11\xrightarrow{}\textcolor{red}d11\xrightarrow{}\textcolor{red}011\ $  (in each state, the number in red indicates the component that changed during the state transition).
In this trajectory, we observed an increase in the activity level of $g_3$ when $g_1$ is at an intermediate level ($\textcolor{red}i00\xrightarrow{}i0\textcolor{red}1$). Hence, in a multivalued refinement with $m_1 = 2$, we proposed a threshold $1$ for the activation from $g_1$ to $g_3$. Additionally, the level of activity of $g_2$ increases when its activator $g_3$ is present and its inhibitor $g_1$ is at an intermediate level ($i0\textcolor{red}1\xrightarrow{}i\textcolor{red}11$). We deduced that $g_1$ needs to be fully active in order to inhibit $g_2$; so we set a threshold at level $2$ for the inhibition from $g_1$ to $g_2$. The resulting proposed multivalued refinement is described in Figure \ref{fig:ToyMulti}A. We verified that its asynchronous dynamics satisfies the  reachability property ${\cal P}_1$, that is the reachability from the state $010$ to the fixed point $011$ (see Figure \ref{fig:ToyMulti}B).

\begin{table}[h!]
    \centering
    \includegraphics[page=3,width=0.7\linewidth]{Figures.pdf}
    \caption{\textbf{Sizes of the Basins of Attraction of the Toy Model Under Different Updating Schemes}. Sizes of the basins of attraction of the fixed points $101$ and $011$ (as a percentage of the size of the Boolean state space) obtained when using either the asynchronous, \mps, or partial \jmps schemes ($J = \{g_1\}$, $\{g_2\}$, or $\{g_3\}$). Red values signal any increase in the size of a basin of attraction when compared to the size of the basin of attraction of the asynchronous dynamics.}
    \label{tab:ToyBasin}
\end{table}

Secondly, we considered a a set of reachability properties ${\cal P}_2$. We observed that the size of the basin of attraction of the fixed point $101$ does not change with the updating scheme, whereas the size of the basin of attraction of the fixed point $011$ changes ($\# \mathcal{B} ^f _{m.p.}(011) > \# \mathcal{B}^f_{asyn}(011)$; see Table \ref{tab:ToyBasin}). As all the asynchronous trajectories are captured in the \mps STG (Section \ref{ssec:m.p.}), this change in size can be interpreted as a loss of reachability in the asynchronous dynamics.
Our aim was to build a refinement restoring all these lost reachabilities (property ${\cal P}_2$).
Using MRBM, we found that we get $\# \mathcal{B}^f_{^Jm.p.}(011) = \# \mathcal{B}^f_{m.p.}(011)$ with the partial \mps scheme with $J=\{g_1\}$. Hence, the asynchronous dynamics of the multivalued refinement defined above allows us to recover a basin of attraction of the same size as for the BM with the \mps dynamics.

\subsection{Application of MRBM to refine a BM of Early Hematopoietic Stem Cell differentiation} \label{sec: HSC}
\begin{figure}[h!]
    \centering
    \includegraphics[page = 4, width=\linewidth]{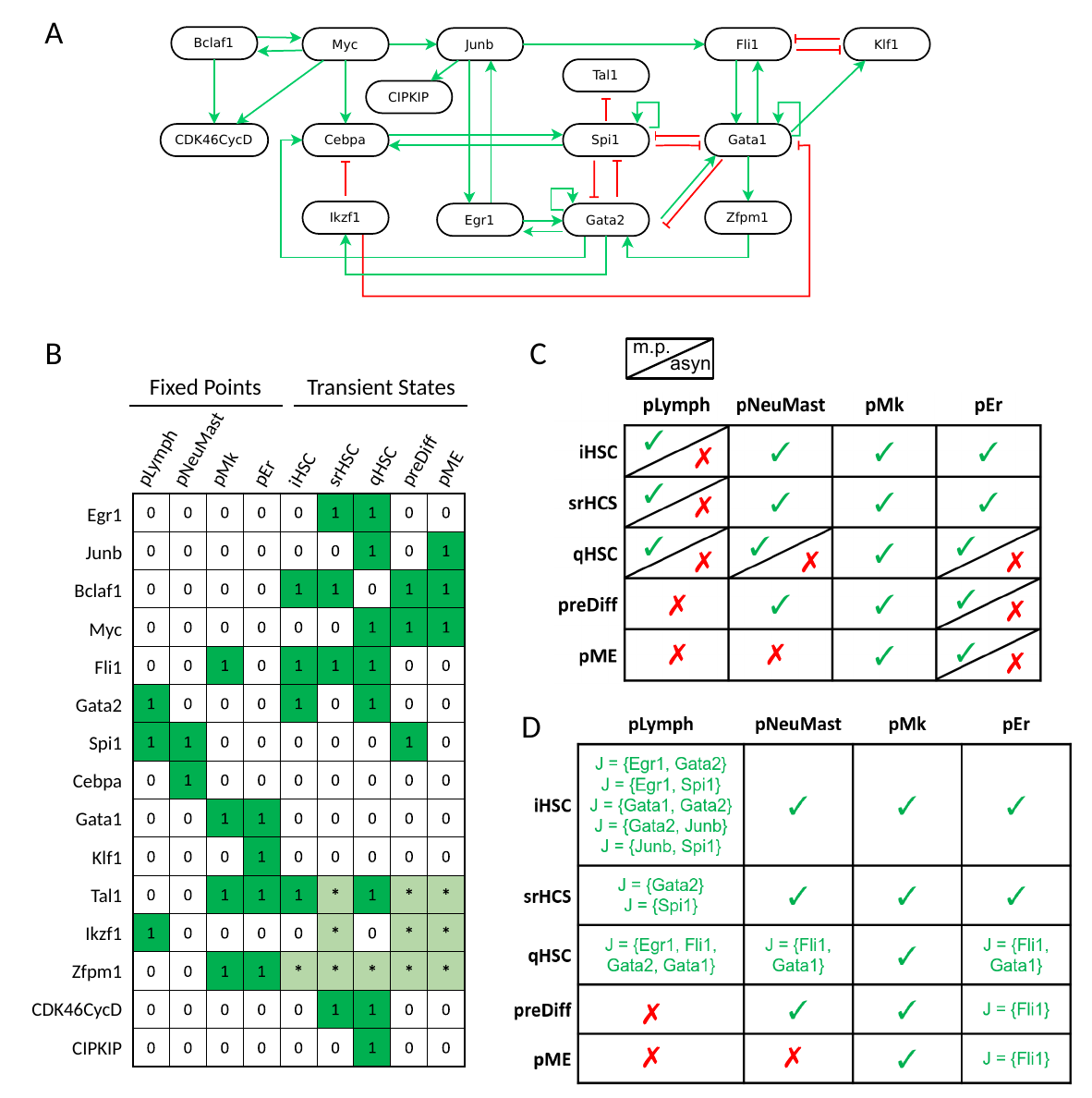}
    \caption{\textbf{Boolean Model of Early Hematopoïetic Stem Cell Differentiation}. (\textbf{A}) Regulatory graph of the BM (edges in red represent inhibitions and edges in green activations). The logical rules are detailed in \cite{herault_novel_2023}. 
    (\textbf{B}) Configuration (i.e. activity level of each component) of the 5 fixed points and the 5  transient key states (in column). Each cell of the table represents the activity level of a component (0=inactive, in white, 1 = active in dark green, * = 0 or 1, in light green). 
    (\textbf{C}) Reachabilities of the fixed points (in column) from transient states (in row) in \mps and asynchronous ("asyn) dynamics. \textcolor{ForestGreen}{\cmark} stands for "reachable", and \textcolor{red}{\xmark} for "not reachable". The cells are divided if there is a difference between the \mps and the asynchronous dynamics. (\textbf{D}) Reachabilities of the fixed points from transient states (in row) in the partial \jmps dynamics. The accessibility of transient states at fixed points in the partial dynamics of \jmps is satisfied for the $J$ sets specified in the table cell.
    }
    \label{fig:HSC}
\end{figure}

The BM depicted in Figure \ref{fig:HSC}A is a model of Hematopoietic Stem Cell (HSC) differentiation detailed in \cite{herault_novel_2023}. 
The model contains 15 components. Its dynamics displays five attractors, all of which are fixed points, each representing a distinct HSC fate (i.e. a distinct final stage of differentiation): lymphoid (pLymph), neutrophils and mastocytes (pNeuMast), erythrocytes (pEr), megakaryocytes (pMk), and an inactive state (zeros). We excluded the fixed point zeros from our study as it was not discussed in \cite{herault_novel_2023}.
The authors of \cite{herault_novel_2023} have defined (from data analysis) five transient key stages in HSC differentiation: initial HSCs (iHSC), self-renewal HSCs (srHSC), quiescent HSCs (qHSC), pre-differentiating HSCs (preDiff), and pME  (denoted "Transient states" in Figure \ref{fig:HSC}B). We tested the fixed points reachabilities from these transient states in the \mps and asynchronous dynamics (Figure \ref{fig:HSC}C). The two dynamics display seven differences in terms of reachability (Figure \ref{fig:HSC}C). 
With the asynchronous scheme, pLymph is not reachable from any of the transient stages of HSC differentiation, pNeuMast is not reachable from qHSC, and pEr is not reachable from qHSCs, pre-diff, and pME. We aimed to identify a multivalued refinement that would provide all those seven reachability properties ($\cal{P}$) under the asynchronous scheme. 

Using the MRBM method, we identified three admissible sets $J$ of size one for which the partial \mps dynamics captures part of the desired reachabilities (Figure \ref{fig:HSC}D):
\begin{itemize}[label=-]
\item With $J=\{Fli1\}$, pEr is reachable from pME and preDiff.
\item With $J=\{Gata2\}$ or $J=\{Spi1\}$, pLymph is reachable from  srHSC.
\end{itemize}
Increasing the size of the sets $J$ by 1, we identified additional sets capturing additional reachabilities of interest:
\begin{itemize}[label=-]
\item With $J=\{Fli1, Gata1\}$, pNeuMast and pEr are reachable from qHSC.
\item With $J=\{Egr1, Gata2\}$, $J=\{Egr1, Spi1\}$, $J=\{Gata1, Gata2\}$, $J=\{Gata2, Junb\}$ or $J=\{Junb, Spi1\}$, pLymph is reachable from the state iHSC.
\end{itemize}
However, pLymph is still not reachable from qHSC. We found that the set $J=\{Egr1,Fli1,Gata2,Gata1\}$ is required for the property $\cal{P}$ to be satisfied in the partial \jmps dynamics.

\begin{figure}[ht!] 
\centering
\includegraphics[width=0.5\linewidth]{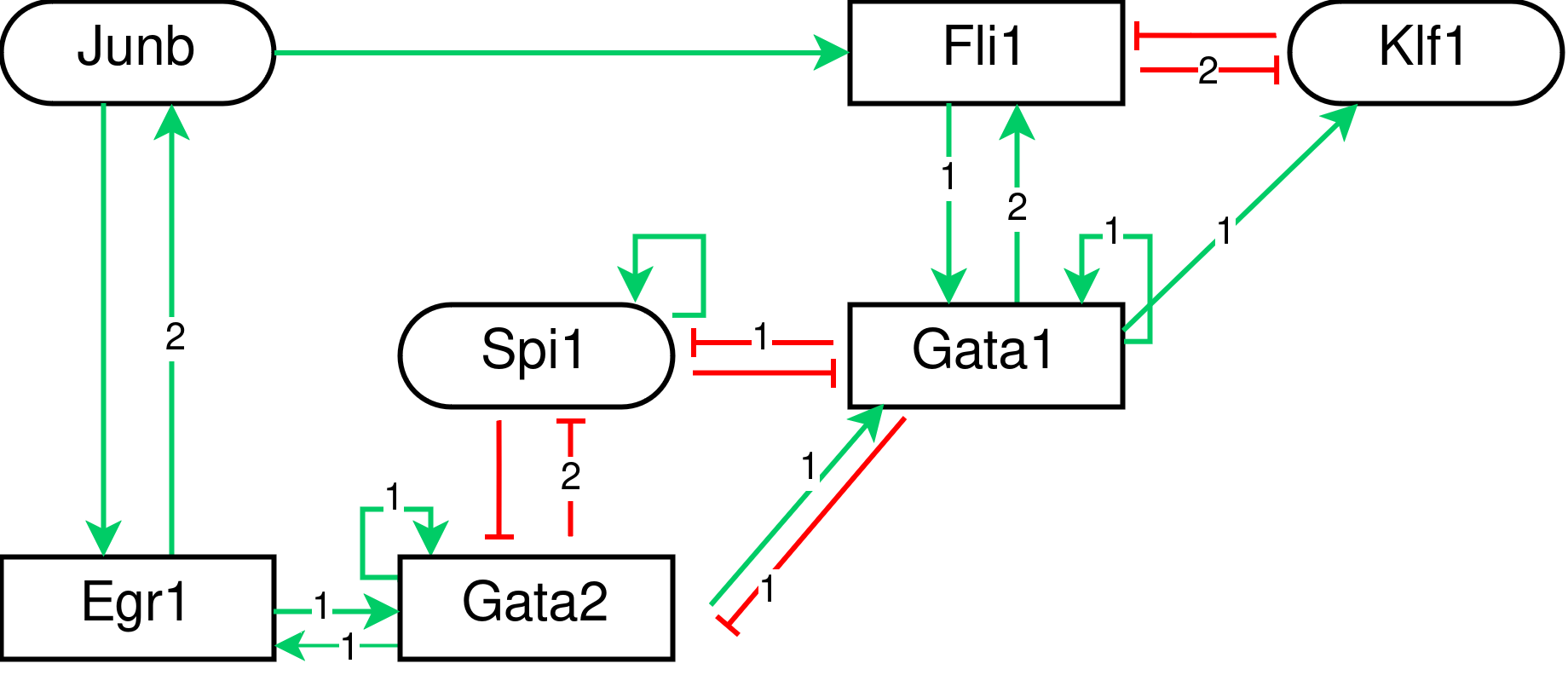}
\caption{\textbf{Subgraph of a Multivalued Refinement of the BM of Early HSC Aging}. Multivalued components are represented as rectangles, while other components are depicted as ovals. The label on arrows indicates the threshold level from which the regulation by a multivalued component occurs.}
\label{fig:paramML}
\end{figure}
We then built a multivalued refinement of the BM that satisfies properties ($\cal{P}$) under the asynchronous scheme. Finding the multivalued functions of the refinement consists in assigning a threshold $1$ or $2$ to each of the outgoing regulations of the 4 multivalued components. First we set to $2$ the maximum activity level of the four components of $J$ ($Egr1,\,Fli1,\,Gata2$, and $Gata1$). 
By analyzing in the partial \jmps dynamics the paths corresponding to the seven reachabilities we want to recover ({\it ad-hoc} strategy, similarly to the path analysis in the case of the Toy model), we decided to set the threshold of 4 regulations at level 2: the inhibition of $Klf1$ by $Fli1$, the inhibition of $Spi1$ by $Gata2$, the activation of $Fli1$ by $Gata1$, and the activation of $Junb$ by $Egr1$ (Figure \ref{fig:paramML}). The logical functions of the multivalued refinement are available in Supplementary Material Section \ref{AnnexD}.

As described in (\cite{herault_novel_2023}), available biological data support part of this parameterisation. Indeed, starting from a pre-differentiated stage (preDiff) cells can differentiate into either neutrophils/mastocytes (pNeuMast), erythrocytes (pEr) or megakaryocytes (pMk). The circuit composed of $Gata1$, $Fli1$, and $Klf1$ controls the cell fate decision between pEr and pMk. Moreover, an analysis of the dynamics of the BM shows that the switch between the pEr fixed point and pMk fixed point from the preDiff state depends on the existence of two distinct thresholds for $Fli1$'s influence on its targets $Klf1$ and $Gata1$ (see \cite{herault_novel_2023} and references therein).

\subsection{Application of MRBM to refine a Boolean Model of Asymmetric Stem Cell Division in \emph{Arabidopsis Thaliana} Root} \label{sec: AT}

As a second case study, we considered the BM of asymmetric stem cell division in the root of \textit{Arabidopsis thaliana} presented in \cite{garcia-gomez_system-level_2020} (Figure \ref{fig:AT}A). The model has six attractors, all of which are fixed points, each of them  corresponding to different cell types observed within the root of the plant, including the quiescent center (QC), peripheral and central pro-vascular initials (PPI and CPI), cortex/endodermis initials (CEI),  columella initials (CoI), and transition domain (TD) (Figure \ref{fig:AT}B).

\begin{figure}[ht!] 
\centering
\includegraphics[page = 5, width=\linewidth]{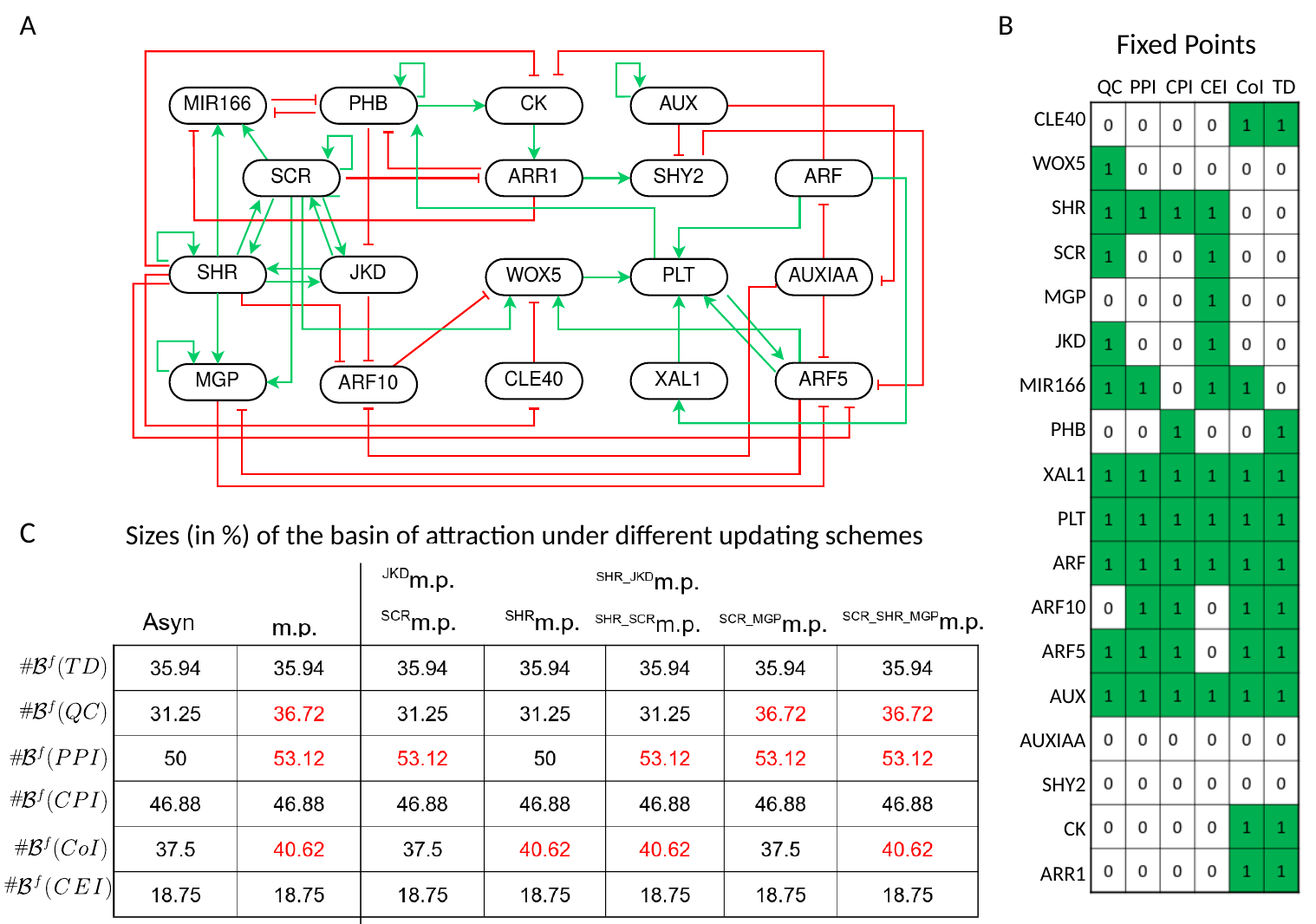}
\caption{\textbf{Boolean model of the Asymmetric Stem Cell Division in \emph{Arabidopsis Thaliana} Root}. 
(\textbf{A}) Regulatory graph of the BM; the logical functions are available in \cite{garcia-gomez_system-level_2020}. 
(\textbf{B}) Summary of the 6 fixed points. Each row represents a component, each column corresponds to a fixed point, and each cells to the activity level ($1$ for active, in dark green, $0$ for inactive, in white). 
(\textbf{C}) Sizes (in percentage) of the basins of attraction of the BM under different updating schemes. The rows represent basins of attraction and the columns correspond to the updating schemes: asynchronous ("Asyn"), \mps and partial \mps. For each partial \jmps dynamics, the contributing set $J$ is indicated in the column headers. We grouped in the same column the partial \jmps dynamics leading to same sizes of basin of attraction. We only report the smallest set(s) of \mps components resulting in a specific size of basin of attraction. Values in red correspond to increases in the sizes of the basin of attraction as compared to the sizes obtained in the asynchronous dynamics.}
\label{fig:AT}
\end{figure}

By comparing the asynchronous dynamics to the \mps dynamics of the model, we noted that the sizes of the basins of attraction of three fixed points (QC, CoI, PPI) are larger within the \mps dynamics than within the asynchronous dynamics (Figure \ref{fig:AT}C). As already mentioned in the Toy model description (Section \ref{TMex}), one can interpret this size difference as a loss of reachability in the asynchronous dynamics. 
To identify a multivalued refinement of the BM whose asynchronous dynamics provides basins of attraction as large as in the \mps dynamics, we applied the MRBM method. Considering sets $J$ of admissible \mps components of size one, we obtained:
\begin{itemize}[label=-]
\item if $J=\{JKD\}$ or $J=\{SCR\}$, then $\#\mathcal{B}^f_{^Jm.p.}(PPI) = \#\mathcal{B}^f_{m.p.}(PPI)$, 
\item if $J=\{SHR\}$, then $\#\mathcal{B}^f_{^Jm.p.}(CoI) = \#\mathcal{B}^f_{m.p.}(CoI)$. 
\end{itemize}
A set $J$ of size two is required for the basin of attraction of the fixed point $QC$:  if $J=\{MGP,SCR\}$, then $\#\mathcal{B}^f_{^Jm.p.}(QC) = \#\mathcal{B}^f_{m.p.}(QC)$. 
Thus, $J=\{SHR,MGP,SCR\}$ is sufficient to ensure identical basin of attraction sizes between the partial \mps and the \mps dynamics (Figure \ref{fig:AT}C). 

Thus, the use of MRBM efficiently restricts the search space for determining a multivalued model. In this example, the use of MRBM enabled us to identify precisely the three components on which it was relevant to search for multi-level parameterisation, without having to test all 816 possibilities (number of 3-element parts of an 18-element set).

The next step is the specification of multivalued refinements. We started looking for multivalued refinements with the maximum value of each multivalued component (that are the components in the previously identified sets $J$) fixed at 2.
We used the exhaustive strategy, testing systematically all possible multivalued  functions, i.e. all possible assignments of the thresholds ($1$ and $2$) to each of the outgoing edges of the multivalued components. We thus provided multivalued refinements which led to expected sizes of asynchronous basin of attraction for all cases except two ($J=\{SCR, MGP\}$ and $J=\{SCR,SHR,MGP\}$). Indeed, for $J=\{SCR, MGP\}$, when the maximum value of the multivalued components is restricted to $2$, only one fixed point among QC and PPI has the same size of basins of attraction in both \mps and asynchronous dynamics (top panel of Figure \ref{fig:paramAT} represents all the possible situations).
We thus increased the maximum value of activity of $SCR$ to $3$. Doing so, we were able to find a refinement which recovers the same sizes for the two basins of attraction (PPI and QC) in both dynamics (Figure \ref{fig:paramAT}, bottom panels). This 3-levels parameterisation of $SCR$ also makes it possible to find a refinement with $J={SCR,SHR,MGP}$ with which we found the same size for the three basins of attraction (PPI, QC and CoI) in both dynamics.

The resulting multivalued refinements are available in the Supplementary Material Section \ref{AnnexD}.

\begin{figure}[ht!] 
\centering
\includegraphics[width=\linewidth]{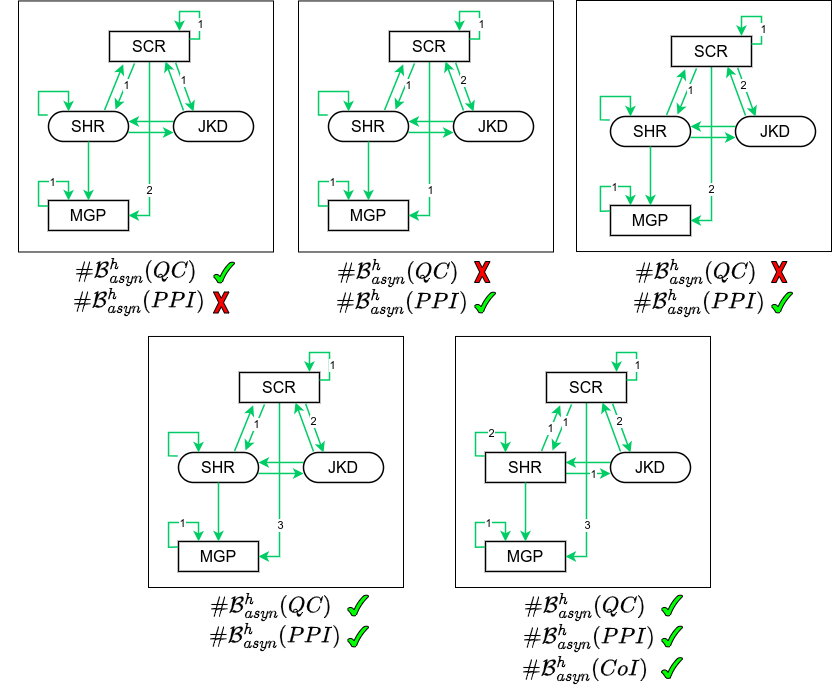}
\caption{\textbf{Subgraph of Multivalued Refinements of the BM of Asymmetric Stem Cell Division in \textit{Arabidopsis thaliana} Root}. Multivalued components are depicted as rectangles, while Boolean components are shown as ovals. The label on arrows indicates the threshold level from which the regulation by a multivalued component occurs. On the top panels $m_{SCR} = 2$, and on the bottom panels $m_{SCR} = 3$. We noted below each panel if we managed to recover the expected size(s) of basin of attraction depending on the set of multivalued components (see Figure \ref{fig:AT}C). The name of the basin of attraction is followed by a checkmark for a success, or a cross for a failure.}
\label{fig:paramAT}
\end{figure}

\section{Discussion}

Although Boolean models have proved to be a valuable and effective tool for modeling biological interaction networks, they have certain limitations due to their high level of abstraction \cite{wang_boolean_2012}. In particular, they capture the salient properties of the dynamics but might overlook other finer-grained properties \cite{chamseddine_hybrid_2020}. Multivalued logical models can be useful to overcome such limitation, as they enable the consideration of more situations (\cite{flobak_discovery_2015}, \cite{selvaggio_silico_2021}). However, the refinement of a Boolean model into a multivalued model is not straightforward. This involves finding the components that need to be multivalued, as well as parameterising these components and their targets. The possibilities are very numerous and we face a combinatorial explosion. When multivaluing a Boolean model, the aim is to capture the expected properties sparingly, by minimising the number of multivalued components and their maximum level.

In this paper, we present MRBM, an original method to refine Boolean models. The aim of these refinements of Boolean models is to encompass specific dynamical properties ${\cal P}$ that appear in the most permissive dynamics but not in the asynchronous Boolean dynamics. We have defined a new updating scheme, the partial most permissive dynamics. This updating scheme applies the most permissive update only to a subset of network components, while asynchronous updating is applied to the other components. 
This allows to identify the components of the network that control the dynamical properties ${\cal P}$, and that would need to be multivalued. Hence, we created a multivalued refinement by multivaluing these components in the Boolean model. Doing so, we induce the dynamical properties ${\cal P}$ in the asynchronous dynamics of a multivalued refinement model. Importantly, the method is supported by mathematical arguments and proofs (see supplementary). The effectiveness of MRBM is demonstrated on a Toy model and two published Boolean models of cell differentiation, highlighting its ability to handle complex biological systems.  



The MRBM method systematically identifies the components to be multivalued but the parameterisation itself remains challenging. We propose  two main strategies: 1) An exhaustive strategy, in which all possible parameterisations are tested and then checked to verify if they satisfy the desired property. This strategy can be time consuming, especially as the maximum activity levels allowed for each components 
increase. 2) An {\it ad hoc} strategy, which analyses precisely relevant paths within the selected partial 
dynamics and thereby provides clues for deciphering the regulation thresholds. 
In both strategies, available a priori biological knowledges are used to guide the definition of the logical functions. 

We applied the MRBM method to two published Boolean models of differentiation: a Boolean model of mammalian early hematopoietic stem cell differentiation and a Boolean model of asymmetric stem cell division in \textit{Arabidopsis thaliana}. In both cases, we have identified reachability properties of biological interest that appear in the most permissive dynamics but not in the asynchronous dynamics. The MRBM method allowed us to successfully identify multivalued refinements of the Boolean models exhibiting these reachability properties under the asynchronous updating scheme.
In the case of the early hematopoietic stem cell differentiation Boolean model, we proposed a multivalued refinement with four multivalued components (with maximal levels of 2). Existing biological data provide support for part of the parameterisation. For the Boolean model representing asymmetric stem cell division in \textit{Arabidopsis thaliana}, we had to multivalue three components, one of them with a maximum activity level set at 3, while the others were set at 2. The maximal level of components is a feature that can be used to identify multivalued nodes requiring more detailed regulatory changes to recover the expected sizes of basins of attraction.
 
There are some limitations for the use of MRBM. 
The method is restricted to dynamics with only stable states (and no cyclic attractors) and focuses on the reachability properties of these stable states. A generalisation to the reachability properties of cyclic attractors is not straightful. Indeed, cyclic attractors depend on the updating scheme. Therefore there may not have a one-to-one correspondance between the cyclic attractors of the most permissive dynamics and those of the asynchronous dynamics. This makes it difficult to carry out the global analysis of the size of the basins of attraction, but a study of specific reachabilities between an initial state and a state belonging to a cyclic attractor remains possible.
In this work, we have forbidden the multivaluation of a self inhibited component. Indeed, the multivaluation of an auto-inhibited component leads to a change in dynamics. It would be necessary, when studying a particular case, to observe this new dynamics in detail before opting for the multivaluation of such a component.

\noindent
The method faces computational issues for large models. We used the python package mpbn to check the dynamical properties. The calculation of the size of the basins of attractions was possible for relatively small models, but when studying larger models (of more than 15 nodes), MRBM method was limited to the study of specific reachability properties.

Finally, we can point out an intrinsic limitation of the method. The number of components to be multivalued (i.e. the size of the set $J$) can be as large as $n$, the number of components. Of course, the larger the set $J$, the more difficult it is to find a parameterisation. It should be noted that difficulties in finding an appropriate refinement may indicate that the initial Boolean model is incomplete, reflecting gaps in our understanding of the system.

MRBM overall provides a systematic and innovative approach to refine Boolean models with multivaluation. It exploits the richness of most permissive dynamics, extracts its significant properties, and integrates them in the asynchronous Boolean model to refine it. Importantly, the MRBM method does not rely on external data sources, making it useful when experimental data are limited. 

\bibliographystyle{unsrtnat}
\bibliography{references} 

\newpage
\appendix

\section*{Appendix}

\section{Construction of Multivalued Refinements\\ of a BM}\label{AnnexA}

Given a BM $f$ on $\mathbb{B}^n$, we recall that given $j_0 \in \{1, \dots,n\}$, the map $f_{j_0}$ can be set in a disjunctive normal form (DNF). In other words, $f_{j_0}(x)$ is written as a disjunction of clauses that are themselves conjunctions of radicals $w\,x_j$ as below:
        \begin{equation}
    \label{DNF}
    f_{j_0}(x)=\underset{h \in \{1,\dots,s_{j_0}\}}\bigvee ~ \underset{k \in \{1,\dots, r_{j_0,h}\}}\bigwedge w_{j_0,h,k}\,x_{j_{j_0,h,k}}\;\;,
\end{equation}
        where operators $\wedge,\, \vee$ stand for {\sc and, or} respectively, and the $w_{j_0,h,k}$ are equal to $\varepsilon$ (empty string), or $\neg$ ({\sc non}).\\
        Moreover, this DNF is supposed to be a shortest one - this ensures that components involved in the expression of $f_{j_0}$ are effective regulators of $g_{j_0}$, and avoids redundancy.

\medskip
 A way to build a refinement $h$ of $f$ on $X = \prod_{j=1}^n\;\{0,1,\dots, m_j\}$ is then the following: 
\begin{itemize}[label = -]
\item Starting from (\ref{DNF}), and the current variable becoming $x\in X$, for each $j=j_{j_0,h,k}$ in (\ref{DNF}) for which $m_j >1$, we choose a partition of $\{0, \dots,m_j\}$ of the form $ \{0, \dots,s\}\cup \{s+1, \dots,m_j\}$ and we replace $w_{j_0,h,k}\,x_{j_{j_0,h,k}}$ with $x_{j_{j_0,h,k}}\geq s+1$ if $w_{j_0,h,k}=\varepsilon$, and with $x_{j_{j_0,h,k}}< s+1$ if $w_{j_0,h,k}=\neg$. Clauses with indices $j$ such that $m_j=1$ remain unchanged.

\item In this way, we obtain a function of the variable $x\in X$ with values in $\B$; let $\cal{H}$$_{j_0}$ be defined by $\cal{H}$$_{j_0}(x)=+1$ when this value is $1$, and $\cal{H}$$_{j_0}(x)=-1$ when it is $0$.

\item Finally, we define the refinement $h=(h_1 \dots,h_n)$ setting for $j\in \{1,\dots,n\}$
\begin{center}
    $h_{j}(x) =max(0,min(m_{j},x_j+$$\cal{H}$$_{j}(x)\, ))$,\end{center}
that ensures that the map $h$ is actually a map from $X$ to itself.

\end{itemize}

\section{Comparison of Reachability Properties in a BM and Multivalued Refinements}\label{AnnexB}

\begin{proposition}  \label{prop refinement}
Let $f$ be a BM on $\B^n$, the map $h$ a refinement of $f$ on $X= \prod_{j=1}^n\;\{0,1,\dots, m_j\}$, and $J=\{j\in \{1, \dots,n\}\;;\; m_j>1\}$. We suppose that no component $g_j$ of $J$ is self-inhibited.
           
Set $a$ and $a'$ two Boolean states of $\{0,1\}^n$, and $b$ and $b'$ the elements of $X$ obtained replacing the coordinate $a_j$ of $a$ (resp. $a'_j$ of $a'$) by $m_j$ when $a_j=1$ (resp. $a'_j=1$), for each $j\in J$.
If there exists a trajectory from $a$ to $a'$ in the asynchronous STG of $f$, then there exists a trajectory from $b$ to $b'$ in the asynchronous STG of $h$.
\end{proposition} 

\begin{proof} 
It is sufficient to prove this property in the case where there is a transition $a\rightarrow a'$ in the asynchronous STG of $f$, and $a_1 \neq a'_1$.
Suppose for instance that $a_1=0$ and $a'_1=1$ (the case $a_1=1$ and $a'_1=0$ can be treated in a similar way). Then the states $b$ and $b'$ differ only by their first coordinates, $b_1=0$ and $b'_1=m_1$.
            
In the previous notation, a conjunctive clause of the form $\underset{k \in \{1,\dots, r\}}\bigwedge w_{k}\,x_{j_{k}}$ in the logical formula of $f_1$ is made true by $a$.
This implies that the corresponding clause of $\Tilde{h}$ built as above is made true by $b$. 

- In the case $m_1=1$, this gives a transition from $b$ to $b'$.

- In the case $m_1>1$,  there is a transition from $b$ to the state $b^{(1)}$ obtained from $b$ by changing $b_1$ into $1$. Considering the hypothesis that $x_1$ does not occur in the clause $\underset{k \in \{1,\dots, r\}}\bigwedge w_{k}\,x_{j_{k}}$, this can be iterated until getting all the transitions $b\rightarrow b^{(1)} \rightarrow \dots \rightarrow b^{(m_1)}=b'$, as expected.
\end{proof} 

\begin{remark}
If a component $g_j$ of $J$ is self-inhibited, this proposition is irrevocably defeated. 
This can be illustrated for instance setting $n=2$, and for $x\in \B ^2$,\\ 
$f_1(x)=\neg x_1 \;\vee \;\neg x_2\,$, $\;\;f_2(x)=\neg x_1$.

- An example of refinement $h$ of $f$ on $X=\{0,1,2\} \times \{0,1\}$ is given by\\ 
$\cal{H}$$_1(x)=+1$ iff $x_1<2\;\vee \;\neg x_2\,$,  $\;\;\;$ $\cal{H}$$_2(x)=+1$ iff $x_1<1\,$.\\
There is a transition from $11$ to $01$ in the asynchronous STG of $f$, but no trajectory from $21$ to $01$ in the asynchronous STG of $h$.

- Another example of refinement $h'$ of $f$ on $X=\{0,1,2\} \times \{0,1\}$ is given by\\ 
$\cal{H}$$'_1(x)=+1$ iff $x_1<1\;\vee \; \neg x_2\,$,  $\;\;\;$ $\cal{H}$$'_2(x)=+1$ iff $x_1<2\,$.\\
The state $10$ is a fixed point of $f$, and the lonely attractor of its asynchronous dynamics. There is a trajectory from $01$ to $10$ in the asynchronous STG of $f$, but no trajectory from $01$ to $20$ in the asynchronous STG of $h'$. Moreover, there is a new attractor in this STG of $h'$, that is a cycle of length $2$ between $01$ and $11$.

\begin{minipage}[t]{.3\linewidth}
\begin{tabular}{c|c}
    $x$&$f(x)$\\
    \hline
     0 0  &   1   1      \\
     0 1  &   1   1     \\
     1 0  &   1   0      \\
     1 1  &   0   0      \\
    
\end{tabular} 
\end{minipage}
\hfill
\begin{minipage}[t]{.3\linewidth}
\begin{tabular}{c|c}
    $x$&$h(x)$\\
    \hline
    0 0   &   1   1      \\
    1 0   &   2   0      \\
    2 0   &   2   0      \\
    0 1   &   1   1      \\
    1 1  &   2   0      \\
    2 1  &   1   0      \\
    
\end{tabular} 
\end{minipage}
\hfill
\begin{minipage}[t]{.3\linewidth}
\begin{tabular}{c|c}
    $x$&$h'(x)$\\
    \hline
    0 0   &   1   1      \\
    1 0   &   2   1      \\
    2 0   &   2   0      \\
    0 1   &   1   1      \\
    1 1  &   0   1      \\
    2 1  &   1   0       \\
\end{tabular} 
\end{minipage}
\end{remark}

\medskip
We are especially interested in the fixed points of the dynamics, for which the way we build refinements gives the following.

        \begin{proposition} \label{fixed points}
          Let $f$ be a BM on $\{0,1\}^n$, the map $h$ a refinement of $f$ on $X= \prod_{j=1}^n\;\{0,1,\dots, m_j\}$, and $J=\{j\in \{1, \dots,n\}\;;\; m_j>1\}$. We suppose that no component $g_j$ of $J$ is self-inhibited.
\begin{itemize}
    \item 

          Set $\omega$ be a fixed point of $f$, and $\omega '$ the element of $X$ obtained replacing the coordinates $\omega_j$ of $\omega$ equal to $1$ by $m_j$, for each $j\in J$. Then $\omega '$ is a fixed point of $h$, and all the fixed points of $h$ are obtained in this way.

    \item    Let suppose that the attractors of the asynchronous dynamics of $f$ are $l$ fixed points $\omega^{(1)}$, ..., $\omega ^{(l)}$. Then the $l$ fixed points $\omega'^{(1)}$, ..., $\omega '^{(l)}$ of $h$ obtained as above are the lonely attractors of the asynchronous dynamics of $h$. 
          \end{itemize}
        \end{proposition}

        \begin{proof} 
        
        \begin{itemize}
    \item The fact that $\omega '$ is a fixed point of $h$ comes from the inequalities required to be a refinement and from $\alpha(\omega ')=\{\omega\}$.
    
    Conversely, suppose that a state $x\in X$ is a fixed point of $h$. Then, by the hypothesis, all the coordinates $x_j$ with $j\in J$ have to be equal to $0$ or $m_j$, and $\alpha(x)$ is reduced to one Boolean state $y$. Finally, the state $x$ being a fixed point of $h$, the state $y$ is necessarily a fixed point of $f$, as expected.
    
    \item Set $x\in X$. By the hypothesis, there is a trajectory from $x$ to a state $z$ whose coordinates $z_j$ with $j\in J$ are all equal to $0$ or $m_j$. Then, $\alpha(z)$ is reduced to one Boolean state $y$. There is a trajectory from $y$ to some fixed point $\omega^{(k)}$ of $f$. By Proposition \ref{prop refinement}, there is thus a trajectory from $z$ to $\omega'^{(k)}$. In conclusion, all the states $x$ of $X$ lead to a fixed point of $h$, that achieves the proof.
       \end{itemize}     
        \end{proof}
        
\bigskip

\section{Generalization of the Completeness Property of the m.p. Scheme to the Partial m.p. Schemes}\label{AnnexC}

For the sake of convenience, in this supplementary, trajectories from a state $x$ to a state $y$ will be denoted in a contracted manner $x \xrightarrow[\text{- - -}]{f}{}^*\;y\;$.

\subsubsection*{Completeness property of the m.p. scheme}

We begin recalling Paulevé $\&$ al theorem ([9]
), and give a detailed constructive proof.

Let us consider an integer $n\geq 1$, a BM $f$ of dimension $n$, and a refinement $h$ of $f$ defined on $X= \prod_{j=1}^n\;\{0,1,\dots, m_j\}$, where $m_j\in \N ^*$ for each $j\in \{1,\dots,n\}$, and at least one of the $m_j$ is $>1$. 
    
For $j\in \{1,\dots,n\}$, we denote by $F_j$ the map defined by $F_j(x)=f_j(x)-x_j$ for $x\in \B^n$, and  by $H_j$ the map defined by $H_j(x)=h_j(x)-x_j$ for $x\in X$ (these maps are with values in $\{-1,0,+1\}$). 

Let $x$ be an element of $X$. Let us call \textit{m.p. state compatible with $x$} any element $\hat{x}$ of $X_{m.p.}=\{0,1,i,d\}^n$ such that for each $j \in \{1,\dots,n\}$,

- if $x_j=0$, then $\hat{x}_j=0$,

- if $x_j=m_j$, then $\hat{x}_j=1$,

- if $x_j\notin \{0,m_j\}$, then $\hat{x}_j=i$ or $d$.

\begin{remark}
In both definitions of the m.p. scheme related to $f$ and of the multivalued refinements of $f$, the associated components are supposed to assume that any intermediate level of some component $g$ (the levels $i$ and $d$, the levels $l$ such that $0<l<m_j$ for some $j$) can be considered as levels where $g$ is, or is not, active. 

For $x\in X$, the set\\ $\alpha(x)= \{x' \in \B ^n \,;\, \forall j\in \{1,\dots,n\},\,(x_j=0\Rightarrow x'_j=0)\,\textrm{and}\,(x_j=m_j\Rightarrow x'_j=1)\}$,\\
       and for $x\in X_{m.p.}$ the set\\ $\gamma (x)=\{x'\in \B ^n\,;\, \forall j \in \{1,\dots,n\}, \,(x_j=0\Rightarrow x'_j=0)\;\textrm{and}\,(x_j=1\Rightarrow x'_j=1)\}$\\
are introduced to this end.

Hence, the proof of the following theorem is essentially based on the fact that if $x\in X$, and $\hat{x}$ is an element of $X_{m.p.}$ compatible with $x$, then $\alpha(x)=\gamma(\hat{x})$.
\end{remark}

\medskip
\begin{theorem} \label{Loïc} In the previous notation, let $x \xrightarrow[asyn]{h}{}^*\;y$ be a trajectory in the asynchronous STG of $h$. For any element $\hat{x}$ of $X_{m.p.}$ compatible with $x$, there exists in the STG of the m.p. dynamics of $f$ a trajectory $\hat{x} \xrightarrow[m.p.]{f}{}^* \;\hat{y}$ such that $\hat{y}$ is compatible with $y$. 
\end{theorem}
     
\begin{proof}   
It is clear that it is sufficient to prove the result in the case where the asynchronous trajectory is reduced to one transition, that is $x \xrightarrow[asyn]{h}{}y$. 

Let $\hat{x}$ be an element of $X_{m.p.}$ compatible with $x$, and $j \in  \{1,\dots,n\}$ the integer such that $x_j\neq y_j$. 

For convenience, we suppose that $j=1$, and we detail the construction of $\hat{x} \xrightarrow[m.p.]{f}{}^* \;\hat{y}\;$ in the case $H_1(x) >0$, the case $H_1(x) <0$ being similar.
\begin{itemize}
\item If $m_1=1$, $x_1=\hat{x}_1=0$ and $y_1=1$, then the existence of $x' \in \beta(x)=\gamma(\hat{x})$ such that $F_1(x')>0$ gives 
 $$\hat{x}=(0,\hat{x} _2, \dots,\hat{x} _n)\; \rightarrow\;(i,\hat{x} _2, \dots,\hat{x} _n)\; \rightarrow\; (1,\hat{x} _2, \dots,\hat{x} _n)\;.$$

\item If $m_1>1$, $x_1=\hat{x}_1=0$ and $y_1=1$, then the existence of $x' \in \beta(x)=\gamma(\hat{x})$ such that $F_1(x')>0$ gives 
$$\hat{x}=(0,\hat{x} _2, \dots,\hat{x} _n)\; \rightarrow\;(i,\hat{x} _2, \dots,\hat{x} _n)\;.$$
 
\item If $m_1>1$, $x_1=l$, where $0<l<m_1-1$ and $y_1=l+1$, then $\hat{x}_1=i \;\textrm{or}\;d$ and $\hat{x}$ is compatible with $y$ : 

\begin{center} stay on $\hat{x}$. \end{center}
\item If $m_1>1$, $x_1=m_1-1$, $y_1=m_1$ and $\hat{x}_1=i$, then the existence of $x' \in \beta(x)=\gamma(\hat{x})$ such that $F_1(x')>0$ gives  
$$\hat{x}=(i,\hat{x} _2, \dots,\hat{x} _n)\; \rightarrow\;(1,\hat{x} _2, \dots,\hat{x} _n)\;.$$
\item If $m_1>1$, $x_1=m_1-1$, $y_1=m_1$ and $\hat{x}_1=d$, then the existence of $x' \in \beta(x)=\gamma(\hat{x})$ such that $F_1(x')>0$ gives  
$$\hat{x}=(d,\hat{x} _2, \dots,\hat{x} _n)\;\rightarrow\;(i,\hat{x} _2, \dots,\hat{x} _n)\; \rightarrow\;(1,\hat{x} _2, \dots,\hat{x} _n)\;.$$
\end{itemize} 
\end{proof}

\subsection*{Completeness property of the partial m.p. schemes}

We consider now an integer $n\geq 1$, a BM $f$ of dimension $n$, a non-empty subset $J$ of $\{1,\dots,n\}$, and a refinement $h$ of $f$ defined on $X= \prod_{j=1}^n\;\{0,1,\dots, m_j\}$, where
$m_j\in \N ^*$ for each $j\in \{1,\dots,n\}$, and $m_j>1$ if and only if $j\in J$. 
    
As above, for $j\in \{1,\dots,n\}$, we denote by $F_j$ the map defined by $F_j(x)=f_j(x)-x_j$ for $x\in \B^n$, and  by $H_j$ the map defined by $H_j(x)=h_j(x)-x_j$ for $x\in X$. 

Let $x$ be an element of $X$. We call \textit{$^J$m.p. state compatible with $x$} any element $\hat{x}$ of $X_{^Jm.p.}$ such that for each $j \in \{1,\dots,n\}$,

- if $x_j=0$, then $\hat{x}_j=0$,

- if $x_j=m_j$, then $\hat{x}_j=1$,

- if $x_j\notin \{0,m_j\}$, then $\hat{x}_j=i$ or $d$.

\begin{theorem} \label{loic generalise}  In the previous notation, let $x \xrightarrow[asyn]{h}{}^*\;y$ be a trajectory in the asynchronous STG of $h$. For any element $\hat{x}$ of $X_{^Jm.p.}$ compatible with $x$, there exists in the STG of the $^J$m.p. dynamics of $f$ a trajectory $\hat{x} \xrightarrow[^Jm.p.]{f}{}^* \;\hat{y}$ such that $\hat{y}$ is compatible with $y$. 
     \end{theorem}
\begin{proof}
The proof is an easy adaptation of the one of Theorem \ref{Loïc}: we restrict ourselves to the case where the considered trajectory is reduced to one transition $x \xrightarrow[asyn]{h}{}y$, with $x_1\neq y_1$ and $H_1(x)>0$. Let $\hat{x}$ be an element of $X_{^Jm.p.}$ compatible with $x$. A suitable trajectory $\hat{x} \xrightarrow[^Jm.p.]{f}{}^* \;\hat{y}$ is obtained on the following way:    
\begin{itemize}
\item If $1\in J$ the construction of $\hat{x} \xrightarrow[^Jm.p.]{f}{}^* \;\hat{y}$ is the same than in the proof of Theorem \ref{Loïc}.
\item If $1\notin J$, then $x_1=0$, $y_1=1$ and the existence of $x' \in \beta(x)=\gamma(\hat{x})$ such that $F_1(x')>0$ gives
$$\hat{x}=(0,\hat{x} _2, \dots,\hat{x} _n)\; \rightarrow\;(1,\hat{x} _2, \dots,\hat{x} _n)\;.$$  
\end{itemize}
\end{proof}

\section{Multivalued refinement of the BM of Early Hematopoietic Stem Cell Aging}\label{AnnexD}

You can find the model in ".bnet" format within the GitHub repository's "Examples/Hérault, Léonard et al (2022)" directory. The entirety of the logical rules for the multivalued refinement presented in the results section is as follows:
\bigskip

\noindent $Egr1 = Gata2 \, \& \,Junb \\
Junb = Egr1:2 \,| \,Myc \\
Blacf1 = Myc \\
Myc = Cebpa \,\& \,Blacf1 \\
Fli1:2 = Junb \,| \,Gata1:2 \,\& \,!Klf1) \\
Gata2:2 = (Gata2 \,\& \,!Gata1 \,\&\, !Zfpm1) \,| \,(Egr1 \,\& \,!Gata1\, \& \,!Zfpm1 \,\& \,!Spi1) \\
Spi1:2 = (Spi1 \,\&\, !Gata1) \,| \,(Cebpa \,\&\, !Gata1\, \& \,!Gata2:2) \\
Cebpa = (Gata2\, \& \,!Ikzf1)\, |\, (Spi1 \,\& \,!Ikzf1) \\
Gata1:2 = Fli1 \,|\, (Gata2\, \& \,!Spi1) \,| \,(Gata1\, \& \,!Ikzf1 \,\& \,!Spi1) \\
Klf1 = Gata1 \,\& \,!Fli1:2 \\
Tal1 = Gata1 \,\& \,!Spi1 \\
Ikzf1 = Gata2 \\
Zfpm1 = Gata1 \\
CDK46CycD = Bclaf1 \,|\, Myc \\
CIPKIP = Junb$
\bigskip

To simplify notation, we denote the activity of a node by its label (e.g., $Gata2$, not $x_{Gata2}$). Additionally, for multivalued components, when the threshold of regulation is other than $1$, it is denoted by ":level". For instance, $Gata2:2$ signifies that $x_{Gata2} = 2$ whereas $Gata2$ means $x_{Gata2} = 1$.

\section{Multivalued refinement of the BM of Asymmetric Stem Cell Division in \emph{Arabidopsis Thaliana} Root}\label{AnnexE}

You can find the multivalued refinement discussed in the corresponding result section within the GitHub repository's "Examples/García-Gómez, Mónica L et al (2020)" directory.

\end{document}